\documentclass[12pt]{article}

\usepackage{datatool} 
\usepackage{graphicx}

\usepackage{algorithm}
\usepackage{algorithmic}

\usepackage{threeparttable}
\usepackage{fullpage}

\usepackage[round]{natbib}

\usepackage{enumerate}
\usepackage{verbatim}
\usepackage{amsfonts,amsmath,amssymb,amsthm}
\usepackage{graphicx,subfigure}
\usepackage{epsf}
\usepackage{epsfig}
\usepackage{fancyhdr}
\usepackage{graphics}
\usepackage{psfrag,latexsym}
\usepackage{url}

\newenvironment{carlist}
 {\begin{list}{$\bullet$}
 {\setlength{\topsep}{0in} \setlength{\partopsep}{0in}
  \setlength{\parsep}{0in} \setlength{\itemsep}{\parskip}
  \setlength{\leftmargin}{0.07in} \setlength{\rightmargin}{0.08in}
  \setlength{\listparindent}{0in} \setlength{\labelwidth}{0.08in}
  \setlength{\labelsep}{0.1in} \setlength{\itemindent}{0in}}}
 {\end{list}}

\newcommand{\bcar}{\begin{carlist}}
\newcommand{\ecar}{\end{carlist}}

\newenvironment{carliste}
 {\begin{list}x
 {\setlength{\topsep}{0in} \setlength{\partopsep}{0in}
  \setlength{\parsep}{0in} \setlength{\itemsep}{\parskip}
  \setlength{\leftmargin}{0.07in} \setlength{\rightmargin}{0.08in}
  \setlength{\listparindent}{0in} \setlength{\labelwidth}{0.08in}
  \setlength{\labelsep}{0.1in} \setlength{\itemindent}{0in}}}
 {\end{list}}

\newcommand{\bcare}{\begin{carliste}}
\newcommand{\ecare}{\end{carliste}}

\newcommand{\realpos}{\ensuremath{\real_+}}

\makeatletter
\long\def\@makecaption#1#2{
        \vskip 0.8ex
        \setbox\@tempboxa\hbox{\small {\bf #1:} #2}
        \parindent 1.5em  
        \dimen0=\hsize
        \advance\dimen0 by -3em
        \ifdim \wd\@tempboxa >\dimen0
                \hbox to \hsize{
                        \parindent 0em
                        \hfil 
                        \parbox{\dimen0}{\def\baselinestretch{0.96}\small
                                {\bf #1.} #2
                                } 
                        \hfil}
        \else \hbox to \hsize{\hfil \box\@tempboxa \hfil}
        \fi
        }
\makeatother

\long\def\comment#1{}

\makeatletter
\def\@cite#1#2{[\if@tempswa #2 \fi #1]}
\makeatother

\makeatletter
\long\def\barenote#1{
    \insert\footins{\footnotesize
    \interlinepenalty\interfootnotelinepenalty 
    \splittopskip\footnotesep
    \splitmaxdepth \dp\strutbox \floatingpenalty \@MM
    \hsize\columnwidth \@parboxrestore
    {\rule{\z@}{\footnotesep}\ignorespaces
      #1\strut}}}
\makeatother

\newcommand{\bit}{\begin{itemize}}
\newcommand{\eit}{\end{itemize}}
\newcommand{\ben}{\begin{enumerate}}
\newcommand{\een}{\end{enumerate}}

\newcommand{\bear}{\begin{eqnarray}}
\newcommand{\eear}{\end{eqnarray}}

\newcommand{\core}{{\operatorname{core}}}

\newcommand{\prob}{\ensuremath{{\mathbb{P}}}}

\newcommand{\ones}{\ensuremath{\mathbf{1}}}
\newcommand{\zeroes}{\ensuremath{\mathbf{0}}}

\newcommand{\beq}{\begin{quotation}}
\newcommand{\enq}{\end{quotation}}

\newcommand{\estart}{\begin{equation}}
\newcommand{\eend}{\end{equation}}

\newcommand{\defn}{\ensuremath{:  =}}

\newcommand{\wvec}{{\bf{w}}}

\newcommand{\Yvec}{{\bf{Y}}}
\newcommand{\Uvec}{{\bf{U}}}

\newcommand{\Wvec}{{\bf{W}}}

\newcommand{\wsca}{{{w}}}

\newcommand{\Ysca}{{{Y}}}

\newcommand{\bec}{\begin{center}}
\newcommand{\enc}{\end{center}}

\newcommand{\beit}{\begin{itemize}}
\newcommand{\enit}{\end{itemize}}

\newcommand{\been}{\begin{enumerate}}
\newcommand{\enen}{\end{enumerate}}

\newcommand{\comsl}{\begin{slide}}
\newcommand{\comspor}{\begin{slide*}}

\newcommand{\comsld}[2]{\begin{slide}[#1,#2]}
\newcommand{\comspord}[2]{\begin{slide*}[#1,#2]}

\newcommand{\mendsl}{\end{slide}}
\newcommand{\mendspo}{\end{slide*}}

\newcommand{\real}{\ensuremath{{\mathbb{R}}}}

\DeclareMathOperator{\sign}{sign}

\DeclareMathOperator*{\argmin}{arg\,min}

\newcommand{\Epsilonvec}{{\bf{E}}}
\newcommand{\eyemat}{\mathbf{I}}

\newcommand{\natbasis}{\mathbf{e}}

\newcommand{\diagmat}{{\textbf{Diag}}}

\newcommand{\vecto}[1]{\textbf{#1}}

\def\bi{\begin{itemize}}
\def\ei{\end{itemize}}
\def\benum{\begin{enumerate}}
\def\eenum{\end{enumerate}}

\def\be{\begin{equation*}}
\def\ee{\end{equation*}}

\def\l{\left}

\def\s{\mathbf{s}}

\newcommand{\Rmat}{\ensuremath{{\textbf{R}}}}
\newcommand{\Rmatexp}{\ensuremath{\text{Poly}(\Rmat,2)}}

\newcommand{\balpha}{\ensuremath{{\boldsymbol{\hat{\theta}^{\lambdavec}}}}}

\newcommand{\betadummy}{\ensuremath{\hat{\beta}^{\text{Dummy}}}}
\newcommand{\betahcadummy}{\ensuremath{\hat{\beta}^{\text{Dummy}}_{\text{hca}}}}
\newcommand{\betaols}{\ensuremath{\hat{\bbeta}^{\text{LS}}}}

\newcommand{\betahcaols}{\ensuremath{\hat{\alpha}^{\text{LS}}_{\text{hca}}}}

\newcommand{\betalamhca}{\ensuremath{\hat{\beta}^{\lambdavec}_{\text{hca}}}}
\newcommand{\betalam}{\ensuremath{\hat{\bbeta}^{\lambdavec}}}
\newcommand{\zzerolam}{\ensuremath{\hat{\beta}^{\lambdavec}_0}}
\newcommand{\zlam}{\ensuremath{\hat{\bz}^{\lambdavec}}}
\newcommand{\zlamrescaled}{\ensuremath{\hat{\bz}^{\lambdavec, \text{rescaled}}}}

\newcommand{\shrink}[1]{\textrm{S}_{#1}}

\newcommand{\pdim}{\ensuremath{p}}
\newcommand{\ddim}{\ensuremath{{d}}}

\newcommand{\numobs}{\ensuremath{n}}

\newcommand{\betastar}{\ensuremath{{\boldsymbol{\beta^*}}}}
\newcommand{\bbeta}{\ensuremath{{\boldsymbol{\beta}}}}
\newcommand{\bz}{{\textbf{z}}}

\newcommand{\Xmat}{\ensuremath{\textbf{X}}}

\newcommand{\betahat}{\ensuremath{\hat{\bbeta}}}

\newcommand{\lambdapm}{\text{SPR}}

\newcommand{\lambdavec}{\ensuremath{\vec{\lambda}}}

\newcommand{\absoluteerr}{\ensuremath{|\hat{E_i}|}}

\newcommand{\bstarhca}{{ {\alpha}_{\text{hca}}^*}}
\newcommand{\bhca}{ {\alpha_{\text{hca}}}}
\newcommand{\weightedloss}{{L_{\text{quadratic}}(\bhca, \bbeta)}}

\newcommand{\Ovec}{{\bf{O}}}

\newtheorem{thm}{Theorem}[section] 
\newtheorem{lem}[thm]{Lemma}

\DTLloaddb{lamvals}{lambdas.csv}

\DTLloaddb{resultsnew}{result820_new.csv}
\DTLloaddb{gamble820}{gamble_result820_new.csv}

\DTLloaddb{gamble410}{gamble_result410_new.csv}
\DTLloaddb{robust205}{result410_new.csv}

\DTLloaddb{weights1}{bamp_box_weights820.csv}
\DTLloaddb{underrated}{underrated.csv}
\DTLloaddb{overrated}{overrated.csv}

\DTLloaddb{bbvalue}{top.csv}

\DTLloaddb{bbvalue2}{top10_ls.csv}

\DTLloaddb{bamptop10}{top10_bampcv.csv}

\title{Penalized Regression Models for the NBA}
\author{Dapo Omidiran \footnotemark}
\date{}

\begin{document}
\begin{filecontents*}{lambdas.csv}
setting,lam1,lam2
$\lambdavec^{820}_{CV}$,$2^{-10}$,$2^{-3}$
"$\lambdavec^{820}_{CV,2}$",$2^{-10}$,$2^{-1}$
$\lambdavec^{410}_{CV}$,$2^{-10}$,$2^{-2}$
\end{filecontents*}

\begin{filecontents*}{result820_new.csv}
metric,Dummy,WLS,Bamp-CV,Bamp-CV2,Ridge
RMSE (millions),547.0447,558.0350,539.7495,541.9725,551.0912
Fraction of games guessed wrong,0.3927,0.3366,0.2854,0.2951,0.3293
Mean of \absoluteerr,10.5394,18.0507,10.5540,11.8109,16.0016
Variance of \absoluteerr,68.5292,187.1446,60.2527,70.7671,147.8993
Median of \absoluteerr,9.4885,15.2577,8.9687,10.2200,13.6127
Min of \absoluteerr,0.0279,0.0346,0.0260,0.1445,0.0109
Max of \absoluteerr,38.8338,79.1965,40.9012,45.9275,72.0122
Empirical $\prob{(\absoluteerr > 1)}$,0.9171,0.9707,0.9463,0.9512,0.9683
Empirical $\prob{(\absoluteerr > 3)}$,0.7683,0.9000,0.8171,0.8732,0.8878
Empirical $\prob{(\absoluteerr > 5)}$,0.6707,0.8366,0.7244,0.7780,0.8122
Empirical $\prob{(\absoluteerr > 10)}$,0.4854,0.6585,0.4415,0.5098,0.6049
\end{filecontents*}

\begin{filecontents*}{gamble_result820_new.csv}
metric,Dummy,WLS,Bamp-CV,Bamp-CV2,Ridge
$\#$ of bets possible,410.0000,410.0000,410.0000,410.0000,410.0000
$\#$ of bets made,263.0000,356.0000,290.0000,321.0000,346.0000
Net $\#$ of bets won,-3.0000,14.0000,42.0000,21.0000,26.0000
Winning percentage,0.4943,0.5197,0.5724,0.5327,0.5376
\end{filecontents*}

\begin{filecontents*}{gamble_result410_new.csv}
metric,Dummy,WLS,Bamp-CV,Bamp-CV2,Ridge
$\#$ of bets possible,820.0000,820.0000,820.0000,820.0000,820.0000
$\#$ of bets made,342.0000,700.0000,503.0000,378.0000,658.0000
Net $\#$ of bets won,-10.0000,6.0000,57.0000,8.0000,14.0000
Winning percentage,0.4854,0.5043,0.5567,0.5106,0.5106
\end{filecontents*}

\begin{filecontents*}{result410_new.csv}
metric,Dummy,WLS,Bamp-CV,Bamp-CV2,Ridge
RMSE (millions),1087.4967,1209.3042,1078.7885,1081.0955,1130.7781
Fraction of games guessed wrong,0.4024,0.4073,0.3049,0.3195,0.3732
Mean of \absoluteerr,10.3326,28.9714,11.4719,10.5780,19.9875
Variance of \absoluteerr,64.9664,524.1851,74.1957,63.1314,238.1755
Median of \absoluteerr,9.2783,23.3688,9.5853,8.8344,17.3077
Min of \absoluteerr,0.0279,0.0491,0.0208,0.0024,0.0158
Max of \absoluteerr,49.2299,150.4097,43.1902,39.9954,98.2374
Empirical $\prob{(\absoluteerr > 1)}$,0.9207,0.9756,0.9573,0.9463,0.9659
Empirical $\prob{(\absoluteerr > 3)}$,0.7768,0.9329,0.8378,0.8341,0.9000
Empirical $\prob{(\absoluteerr > 5)}$,0.6744,0.8817,0.7378,0.7207,0.8195
Empirical $\prob{(\absoluteerr > 10)}$,0.4720,0.7805,0.4780,0.4524,0.6890
\end{filecontents*}

\begin{filecontents*}{bamp_box_weights820.csv}
Statistic,Descriptions,Weight,LeBron
2M,Per 36 Minute,3.38,7.83
2A,Per 36 Minute,-1.54,14.16
3M,Per 36 Minute,1.48,1.08
3A,Per 36 Minute,-0.21,3.28
FTM,Per 36 Minute,0.73,5.91
FTA,Per 36 Minute,-0.33,7.79
OR,Per 36 Minute,0.11,0.94
DR,Per 36 Minute,0.50,5.99
AS,Per 36 Minute,0.85,6.51
ST,Per 36 Minute,1.66,1.46
TO,Per 36 Minute,-1.88,3.34
BK,Per 36 Minute,0.86,0.58
PF,Per 36 Minute,-0.37,1.92
TC,Per 36 Minute,2.81,0.09
DQ,Per 36 Minute,6.98,0.00
P1,Boolean,-0.17,0.00
P2,Boolean,-0.71,0.00
P3,Boolean,0.29,1.00
P4,Boolean,1.65,0.00
\end{filecontents*}

\begin{filecontents*}{underrated.csv}
Name,Poss,LS,RR,BAMPCV,BAMPCVBoxScore,BAMPCV2,Discrepancy
"Dooling, Keyon",785.0000,12.5121,11.0960,1.3531,-0.3840,-0.1377,1.7371
"Watson, Earl",547.0000,18.4760,13.7020,0.9973,-0.2920,3.1021,1.2893
"Aldridge, LaMarcus",1254.0000,6.5440,8.6985,5.0226,3.8456,5.3833,1.1770
"Ginobili, Manu",1000.0000,13.9528,12.7020,5.4336,4.2599,5.8928,1.1738
"Tolliver, Anthony",583.0000,13.1848,11.5847,1.9424,0.7742,2.8535,1.1682
"Bosh, Chris",993.0000,13.0587,11.4294,4.6311,3.4681,3.8466,1.1630
"Carter, Vince",705.0000,6.9126,6.0004,1.4315,0.3376,0.3924,1.0939
"Collins, Jason",202.0000,13.7382,12.0229,-2.4071,-3.4634,-3.8926,1.0563
"Hill, George",1056.0000,10.9259,9.5845,2.0535,1.0201,2.4380,1.0333
"Bass, Brandon",730.0000,7.4450,7.2850,2.5777,1.5792,1.5300,0.9985
\end{filecontents*}

\begin{filecontents*}{overrated.csv}
Name,Poss,LS,RR,BAMPCV,BAMPCVBoxScore,BAMPCV2,Discrepancy
"Dragic, Goran",559.0000,-7.5569,-9.7152,-2.1439,-0.4758,-3.6124,-1.6680
"Marion, Shawn",936.0000,-11.1210,-9.4488,0.7049,2.2120,2.4033,-1.5071
"Gortat, Marcin",655.0000,-9.4545,-8.0015,2.7433,4.1681,1.9631,-1.4249
"Ellis, Monta",1262.0000,-7.3929,-6.9847,-0.1175,1.2634,1.5922,-1.3810
"Biedrins, Andris",530.0000,-2.0536,-3.5590,0.9205,2.2408,-5.5809,-1.3203
"Jefferson, Al",981.0000,-16.1087,-11.0281,1.7706,3.0598,-0.3927,-1.2893
"Felton, Raymond",1174.0000,-7.7758,-6.8096,1.4721,2.7513,0.1502,-1.2792
"Bell, Raja",772.0000,-5.8359,-7.7140,-2.8779,-1.5992,-0.7038,-1.2787
"Dudley, Jared",865.0000,3.2730,2.0116,1.1287,2.3858,1.3208,-1.2572
"Law, Acie",348.0000,-13.4168,-12.9659,-1.7759,-0.6247,-4.6449,-1.1512
\end{filecontents*}

\begin{filecontents*}{top.csv}
Rank,Player,Rating1,Name,Rating2
1,"James, LeBron",12.62,"Howard, Dwight",10.14
2,"Durant, Kevin",11.5,"Paul, Chris",9.87
3,"Nash, Steve",11.39,"Nowitzki, Dirk",9.47
4,"Paul, Chris ",10.42,"Garnett, Kevin",8.01
5,"Nowitzki, Dirk ",10.33,"Ginobili, Manu",7.85
6,"Collison, Nick ",9.87,"James, LeBron",7.79
7,"Wade, Dwyane ",9.59,"Hilario, Nene",7.65
8,"Hilario, Nene  ",8.56,"Chandler, Tyson",6.84
9,"Deng, Luol ",8.46,"Pierce, Paul",6.75
10,"Howard, Dwight ",8.13,"Wade, Dwyane",6.42
\end{filecontents*}

\begin{filecontents*}{top10_ls.csv}
Name,Poss,LS,RR,BAMPCV,BAMPCVBoxScore,BAMPCV2,Discrepancy
"Foster, Jeff",418.0000,34.2967,19.7560,2.4796,1.5379,11.1418,0.9417
"Chandler, Tyson",708.0000,27.6789,14.8283,5.5826,5.1495,5.8568,0.4331
"Jones, Dahntay",134.0000,26.4939,15.9998,-0.4993,-1.0268,2.0943,0.5275
"Hibbert, Roy",718.0000,23.1868,10.1534,-0.7857,-1.2968,3.9383,0.5112
"Williams, Jason",151.0000,21.7035,15.8257,2.0246,1.7325,5.1953,0.2922
"Mahinmi, Ian",183.0000,21.1010,8.4206,0.7594,0.6082,5.7651,0.1512
"Okafor, Emeka",838.0000,19.4840,14.6781,2.5442,1.8892,2.4739,0.6550
"McRoberts, Josh",495.0000,19.3572,8.6386,5.3526,5.0687,7.0247,0.2839
"Watson, Earl",547.0000,18.4760,13.7020,0.9973,-0.2920,3.1021,1.2893
"Jones, Solomon",345.0000,18.3961,5.8150,-4.1722,-4.3806,-1.6175,0.2085
\end{filecontents*}

\begin{filecontents*}{top10_bampcv.csv}
Name,Poss,LS,RR,BAMPCV,BAMPCVBoxScore,BAMPCV2,Discrepancy
"James, LeBron",1216.0000,15.0150,12.3485,8.6006,7.7318,10.8341,0.8687
"Garnett, Kevin",755.0000,12.1674,10.4794,8.2860,7.9087,9.5917,0.3773
"Paul, Chris",1211.0000,5.8410,7.0654,8.1899,7.6380,9.8212,0.5519
"Nowitzki, Dirk",786.0000,12.4655,12.7649,7.7296,6.7762,12.2978,0.9534
"Howard, Dwight",1148.0000,9.9220,9.4897,7.4706,6.5626,12.7853,0.9079
"Gasol, Pau",1088.0000,6.9671,7.0095,6.9251,6.6350,9.2493,0.2900
"Odom, Lamar",1036.0000,8.2079,7.7147,6.5630,6.3011,7.8508,0.2619
"Hilario, Nene",824.0000,7.2802,6.5632,6.2170,6.4367,8.7966,-0.2197
"Evans, Jeremy",137.0000,9.4967,8.8399,6.1725,5.8435,3.5583,0.3291
"Nash, Steve",893.0000,18.0405,14.3191,6.0349,5.4753,10.2920,0.5596
\end{filecontents*}

\maketitle

\footnotetext[1]{
Dapo Omidiran is at the University of California, Berkeley, CA 94720 (e-mail: dapo@eecs.berkeley.edu).} 

\begin{abstract}

In the National Basketball Association (NBA), teams must make choices about which players to acquire, how much to pay them, and other decisions that are fundamentally dependent on player effectiveness. Thus, there is great interest in quantitatively understanding the impact of each player.
In this paper we develop a new penalized regression model for the NBA, use cross-validation to select its tuning parameters, and then use it to produce ratings of player ability. We then apply the model to the 2010-2011 NBA season to predict the outcome of games. We compare the performance of our procedure to other known regression techniques for this problem, and demonstrate empirically that our model produces substantially better predictions. We evaluate the performance of our procedure against the Las Vegas gambling lines, and show that with a sufficiently large number of games to train on our model outperforms those lines.
Finally, we demonstrate how the technique developed in this paper can be used to quantitively identify ``overrated" players who are less impactful than common wisdom might suggest.
\end{abstract}

{\bf Keywords:} Basketball; Penalized Regression; Ridge Regression; Lasso; Convex Programming

\tableofcontents

\section{Introduction}
The National Basketball Association (NBA) is a multi-billion dollar business. Each of the thirty franchises in the NBA try their best to put forward the most competitive team possible within their budget. To accomplish this goal, a key task is to understand how good players are.

A large fraction of the thirty NBA teams have quantitative groups analyzing data to evaluate and rate players. The website ESPN.com has many analysts providing statistical analysis for casual fans. Gambling houses use quantitative analysis to price bets on games, while gamblers try to use quantitative analysis to find attractive wagers.

A popular technique for producing player ratings is weighted least-squares (LS) regression\footnote{this technique is also known as Adjusted Plus/Minus (APM) in the quantitative basketball community.}. However, as we show later show, least squares is an approach with many flaws.

In this paper, we introduce a new penalized regression technique for estimating player ratings which we call Subspace Prior Regression (henceforth, $\lambdapm$). $\lambdapm$ corrects some of the flaws of least squares for this problem setting, and has substantially better out-of-sample predictive performance. Furthermore, given sufficient training data $\lambdapm$ outperforms the Las Vegas wagering lines.

We interpret the ratings produced by $\lambdapm$, discussing it identifies as the best players in the NBA (Section \ref{best_players}),
who are the most overrated and underrated players (Section \ref{underrated_analysis}),
and what $\lambdapm$ suggests is the relative importance of different basic actions within the game like three point shooting and turnovers
(Section \ref{bsweights_interpretation}).
Finally, we discuss some possible improvements to this model (Section \ref{future}).

\section{Notation}
We use the notation $\realpos^{d}$ to indicate the set
$\{ x \in \real^{d} \text{ } | x_i >0 \text{ } \forall i \}.$
Let $\ones_\numobs \in \real^{n \times 1}$
denote the column vector of ones and
$\natbasis_i \in \real^{n \times 1}$ signify the $i^{th}$ standard basis vector.
We use
$\eyemat_\pdim \in \real^{\pdim \times \pdim}$
to denote the identity matrix of size $\pdim$,
and $\diagmat(\vecto{w})$ to stand for a diagonal matrix with entries given by the vector $\vecto{w}$.
Given $\vecto{a}, \vecto{b} \in \real^{\numobs}$ and $\vecto{c} \in \realpos^n$ we define the inner product as
$
\vecto{a}^T \vecto{b} \defn \sum_{i=1}^{\numobs} a_i b_i,
$
the $\ell_p$ norm
$
||\vecto{a}||_{p} \defn \left[{\sum_{i=1}^{\numobs} a_i^p}\right]^{\frac{1}{p}}
$
and finally the $\vecto{c}$-weighted $\ell_p$ norm as
$
||\vecto{a}||_{p,\vecto{c}} \defn \left[{\sum_{i=1}^{\numobs} c_i a_i^p}\right]^{\frac{1}{p}}.
$

\section{A Brief Introduction to the Game of Basketball}
Each of the thirty teams in the NBA plays 82 games in a season, where 41 of these games are at their home arena and 41 are played away. Thus, there are 1,230 total games in an NBA regular season. Each team has a roster of roughly twelve to fifteen players. Games are usually 48 minutes long,
and each of the two competing teams has exactly five players on the floor at a time. Thus, there are ten players on the floor for the duration of the game.
Associated with each game is a box score, which records the statistics of the players who played in that game.
Figure \ref{boxscore} contains a sample box score from an NBA game played on February 2nd, 2011 by the Dallas Mavericks
 (the home team) against the New York Knicks (the away team). Note that we only display the box score for the Mavericks players.
Observe that there are 12 players listed in the box score, but only 11 who actually played for the Mavericks in this game. Each of the columns of this box score corresponds to a basic statistic of interest (the column REB
in the box score denotes rebounds, AST denotes steals, etc.)

\begin{figure}[!ht]
\caption{Sample single-game boxscore for the Dallas Mavericks}
\centering
\includegraphics[scale=.5]{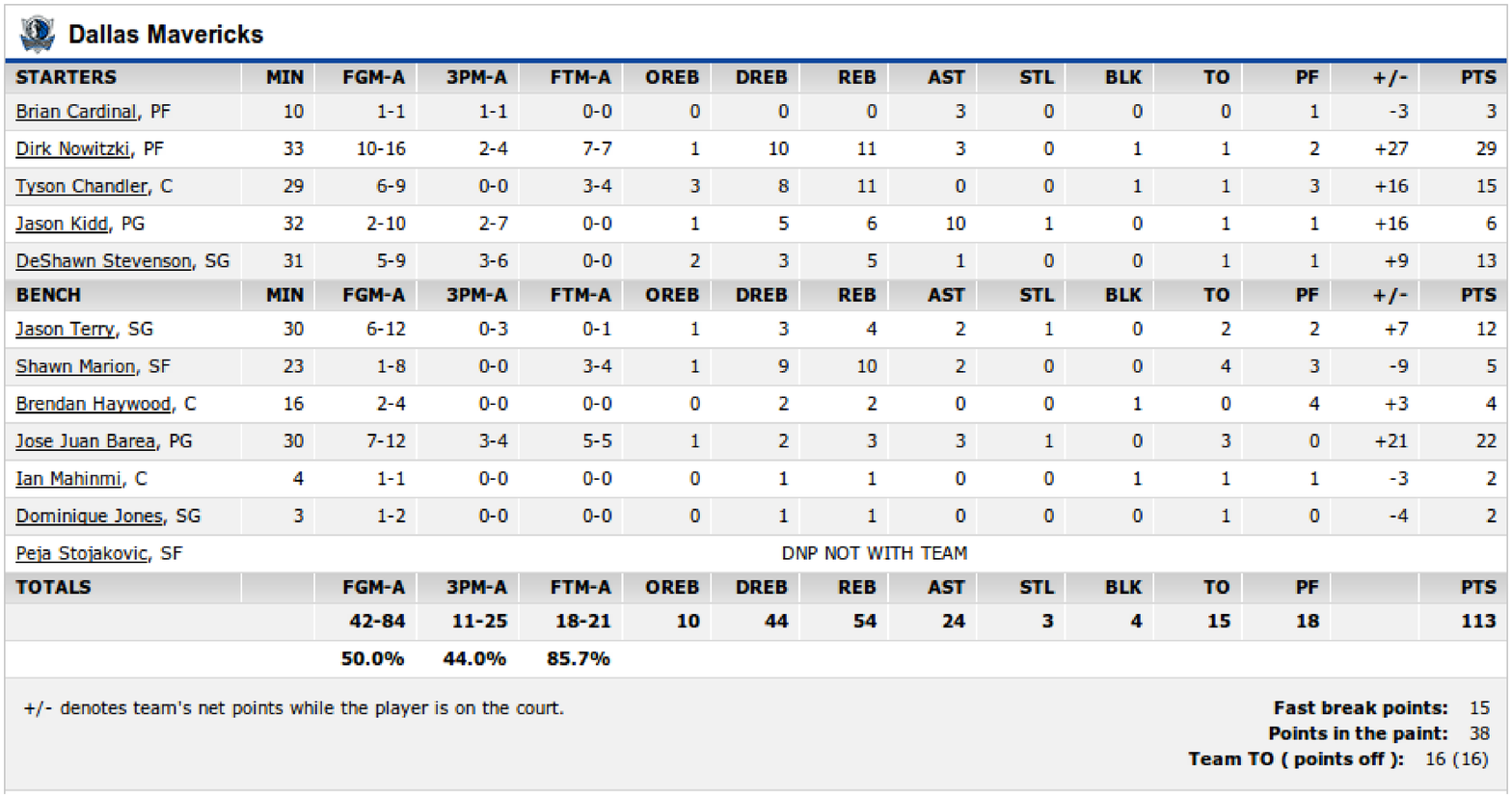}
\label{boxscore}
\end{figure}

\subsection{Statistical Modeling of Basketball}
To statistically model the NBA, we must first extract from each game a dataset suitable for quantitative analysis.
There is a standard
procedure for this currently used by many basketball analysts \citep{Kubatko07, Oliver2004}, which we describe as follows.

We model each basketball game as a sequence of $\numobs$ distinct events between two teams. During event $i$ the home team scores $\Ysca_i$ more points than the away team.
We use the variable $\pdim$ to denote the total number of players in the league (in
a typical NBA season, $\pdim \approx 450$.) We can then represent the current players on the floor for event $i$ with a vector $\Xmat_{i} \in \real^{\pdim}$ defined as

\begin{align*}
\Xmat_{ij} =
\begin{cases}
$1$ & \mbox{Player $j$ is on the floor for the home team} \\
$-1$ & \mbox{Player $j$ is on the floor for the away team} \\
0 & \mbox{otherwise.}
\end{cases}
\end{align*}

Associated with event $i$ is a weighting factor $\wsca_i$. Roughly speaking, the $i^{\text{th}}$ event happens for $\wsca_i$ minutes.

Figure \ref{boxscore} contains a sample box score.
We summarize box score data like that of Figure \ref{boxscore} with the matrix $\Rmat_{\text{Mavericks},\text{Game 1}}$ which looks like
\[
 \Rmat_{\text{Game \#1}}^{\text{Mavericks}} =
\bordermatrix{\text{}&\text{MIN}&\text{FGM}&\text{FGA}&\ldots &\text{PTS}\cr
                \text{Brian Cardinal}&10 &  {1} &{1}  & \ldots & 3\cr
                \text{Dirk Nowitzki}& 33  &  {10} & {16} & \ldots & 29\cr
                & \vdots & \vdots & \vdots &\ddots & \vdots\cr
                \text{Peja Stojakovic}& 0  &  0 & 0       &\ldots & 0}
\]

This matrix records the statistics of the 12 players on the Dallas Mavericks roster for that particular game.
If there are $\ddim$ basic statistics of interest in this box score, then $\Rmat_{\text{Game \#1}}^{\text{Mavericks}}$
is a matrix of size $12$ by $\ddim$.

One can imagine computing the aggregate box score matrix
\begin{align*}
\Rmat^{\text{Mavericks}} = \sum_{t=1}^{82}   \Rmat_{\text{Game \#t}}^{\text{Mavericks}}
\end{align*}

that summarizes the total statistics of these 12 players for an entire season.
Finally, define the $\pdim \times \ddim$ matrix $\Rmat$ that vertically concatenates $\Rmat_{j}$ across the $30$ teams in the NBA:
\[
 \Rmat \defn
\bordermatrix{\text{}&\text{}\cr
                \text{Team 1}&\Rmat^{\text{Mavericks}}\cr
                \text{Team 2}&\Rmat^{\text{Bulls}}\cr
                \vdots & \vdots\cr
                \text{Team 30}&\Rmat^{\text{Celtics}}
                }.
\]

$\Rmat$ summarizes the season box score statistics for all $\pdim$ players who played in the NBA for that year.


\subsection{Least Squares Estimation}

We want to determine the relationship between $\Xmat_i$ and $\Yvec_i$, i.e., find a function $f$ such that $\Yvec_i \approx f(\Xmat_i)$. One natural way to do this is through a linear regression model,
which assumes that

\begin{align*}
\Yvec_i = \bstarhca + \Xmat_i^{T} \betastar + e_i, i=1, 2, \ldots \numobs.
\end{align*}

Recall that the event $i$ has a weighting factor $\wsca_i$ associated with it. Roughly speaking, event $i$
happens for $\wsca_i$ minutes.
 
The scalar variable $\bstarhca$ represents a home court advantage term,
while the variable $\betastar \in \real^{\pdim}$ is interpreted as the number of points each of the $\pdim$ players in
the league ``produces" per minute.
This model recognizes players for whom their team is more effective because of their presence on the floor.

For notational convenience, we stack the variables $\Yvec_i$, $\wsca_i$, and $e_i$ into the $\numobs$ vectors $\Yvec$, $\Wvec$, and $\Epsilonvec$ and the variables $\Xmat_i$ into the $\numobs \times \pdim$ matrix $\Xmat$.
This yields the matrix expression
\begin{align}
\label{apm_model}
\Yvec = \ones_\numobs \bstarhca + \Xmat\betastar + \Epsilonvec.
\end{align}

Given observations $(\Yvec, \Xmat)$ and weights $\Wvec$, we define the $\Wvec$-weighted quadratic loss function as
\begin{align}
\label{quadraticloss}
L_{\text{quadratic}}(\bhca, \bbeta) \defn \frac{1}{\sum_{i=1}^\numobs \wsca_i} ||\Yvec-\ones_\numobs \bhca - \Xmat\bbeta||_{\Wvec}^2.
\end{align}

A natural technique for estimating the variables $\bstarhca$ and $\betastar$ is
to minimize \eqref{quadraticloss}, i.e.,
\begin{align}
\label{apm_algorithm}
(\betahcaols, \betaols) = \arg\min_{\bhca, \bbeta} L_{\text{quadratic}}(\bhca, \bbeta),
\end{align}

resulting in a weighted least squares (LS) problem.
The values $\betaols$ are known in the quantitative basketball community as the \emph{adjusted plus/minus ratings}\footnote{\url{http://www.82games.com/ilardi1.htm}}. The website Basketballvalue\footnote{\url{http://www.basketballvalue.com}} has computed $\betaols$ for several recent seasons.

\subsection{Is least squares regression a good estimator of player value?}
\label{apm_sucks}

\begin{table}[!ht]
  \centering
  \begin{threeparttable}
    \caption{LS Player Ratings}
      \begin{tabular}{l*{5}{c}r}
                    Rank & Player & $\betaols_{i}$\\
        \hline
        \DTLforeach{bbvalue}{%
        \firstname=Rank,\surname=Player,\score=Rating1}{%
        \firstname & \surname & \score\\}
      \end{tabular}
    \label{tab:my_table}
  \end{threeparttable}
\end{table}


Table \ref{tab:my_table} lists the top ten players in the NBA for the combined 2009-2010 and 2010-2011 NBA regular seasons by their ratings produced from least squares\footnote{These numbers were obtained from \url{http://basketballvalue.com/topplayers.php?&year=2010-2011}}. By this ranking, LeBron James was the best player in the league
over this two year period. Since $\betaols_{\text{LeBron James}}=12.62$, this procedure suggests that he is worth an additional 12.62 net points to his team for every 100 possessions the team plays.

How believable are the player ratings of Table \ref{tab:my_table}? The list has many of the widely-considered best players in the NBA. However, there are also some names on this list
that are questionable. If we believe these ratings, then Nick Collison, a player considered by most
fans and analysts to be at best a merely average player at his position, is better than Dywane Wade and Dwight Howard, two of the premiere superstars in the league. Similarly, while Nene Hilario
and Luol Deng are good players, they are not considered by most fans and analysts to be amongst the top ten players in the NBA.

This contradiction between common wisdom and least squares is useful, since it can either reveal to us
that the common wisdom is wrong or that the least squares approach is incorrect.
We need some basis of comparison to evaluate how well least squares is performing.

In classical linear regression, assuming that the generative model satisfies certain conditions, the least squares estimate has several desirable properties (maximum likelihood estimate, best linear unbiased estimate, consistency, asymptotic normality, etc).
However, these properties typically assume that the underlying model
satisfies certain technical conditions like normality, linearity, and statistical independence.
It is unreasonable to expect that these technical conditions hold for the game
of basketball. Thus, we must
find other ways to evaluate how trustworthy the $\betaols$ values are, and whether they should be believed over common wisdom about players.
One simple approach for evaluating the the least squares model is to test
its predictive power versus a simple dummy estimator.

To do this, we
\begin{enumerate}
\item define a dummy estimator that sets $\betadummy_i=0$ for each player,
and the home court advantage term $\betahcadummy=3.5$. In other words, each player is rated a zero, and the home team is predicted to win every
100 possessions by $3.5$ points.
\item We then can compute both the least squares estimate and dummy estimate for the first 820 games of an NBA season, and measure how well each technique does in estimating the margin of victory of the home team for the remaining 410 games of that season.
\end{enumerate}

If least squares accurately models the NBA, then at a minimum it must substantially outperform the dummy estimator.
Let us use the variable $A_k$ to denote the \emph{actual} number of points by which the home team wins game $k$, $\hat{A}_k$ to denote the \emph{predicted} number of points by the statistical estimator of interest,  and $\hat{E_k} \defn \hat{A}_k - A_k$ to denote the error this statistical estimator makes in predicting the outcome of game $k$.

Figure \ref{figdump} is a histogram of the error variable $\hat{E_k}$ over the course of the 410 games under consideration from the 2010-2011 NBA season for each technique. A perfect estimator would have a spike of height $410$ centered around zero. Thus, the ``spikier" the histogram looks, the better a method performs. It is hard to immediately say from Figure \ref{figdump} that the least squares estimate yields better predictions than the simple dummy estimate. We can also study some of the empirical properties of $\hat{E_k}$ for each approach. Table \ref{tab:hca_vs_apm} summarizes the results.

\begin{table}[!ht]
    \centering
    \caption{Performance of Statistical Estimators over the last 410 games}
  \begin{threeparttable}
        \begin{tabular}{l*{6}{c}r}
                Metric  & Dummy & LS & RR & \lambdapm & \lambdapm 2 \\
          \hline
          \DTLforeach{resultsnew}{%
          \firstname=metric,\surname=Dummy,\score=WLS,\lpm=Bamp-CV,\lpmtwo=Bamp-CV2,\ridge=Ridge}{%
          \firstname & \surname & \score & \ridge & \lpm & \lpmtwo\\}
        \end{tabular}
  \end{threeparttable}
    \label{tab:hca_vs_apm}

\end{table}

When comparing least squares to the dummy estimate, we notice that
\begin{enumerate}
\item least squares reduces the percentage of games in which the wrong winner is identified from $39.27\%$ to $33.66\%$ over the block
of 410 games of interest.
\item Unfortunately, the empirical behavior of $\hat{E_k}$ seems to be substantially worse for least squares. For example, the empirical mean of $\absoluteerr$ is $18.05$ for least squares, while only $10.54$ for the dummy estimator. Thus, least squares makes larger average errors
when predicting the final margin of victory of games.
\end{enumerate}

As a result, it is hard to convincingly argue that least squares approach is a better model for the NBA than the dummy estimate.

\begin{figure}[!ht]
\centering
\includegraphics[scale=0.45]{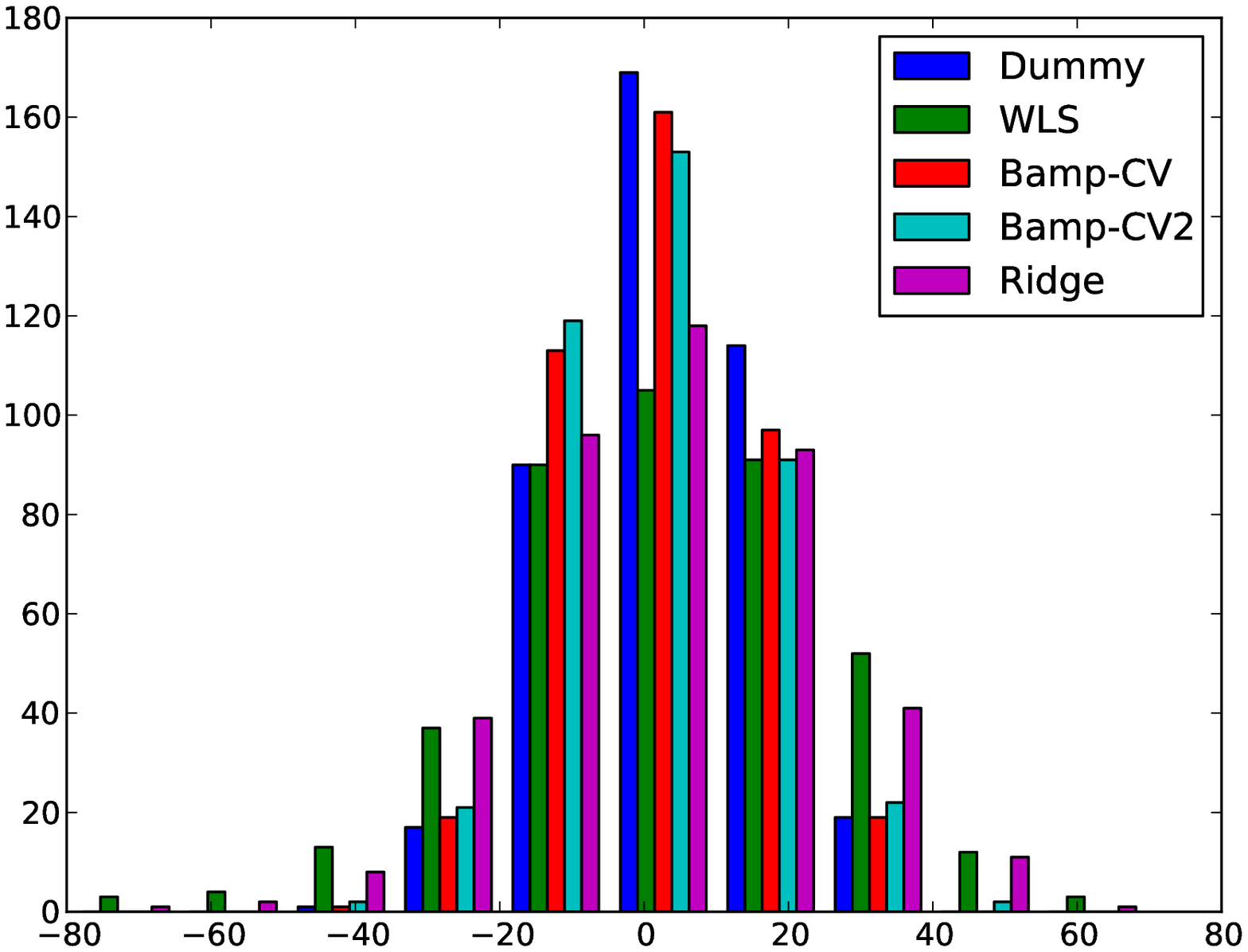}
\caption{Comparison of Dummy, Least Squares, Ridge Regression, $\lambdapm$ and $\lambdapm 2$ trained on 820 games.}
\label{figdump}
\end{figure}

\section{$\lambdapm$: Improving Least Squares}
\label{lamp_rocks}
Although Figure \ref{figdump} and Table \ref{tab:hca_vs_apm} suggest that the LS estimate performs poorly, this doesn't necessarily mean that the linear model \eqref{apm_model} is without promise. The least squares estimate simply doesn't take into account the following two key pieces of information we have about the problem domain:

\begin{enumerate}
\item Model sparsity: The NBA is a game dominated by star players. Lesser players have far less impact on wins and losses. This folk wisdom informs player acquisitions and salaries. For example, with a \$60 million
budget, one would much rather acquire three elite \$15 million stars and fill out the rest of the roster with cheap role-players, than spend tons of money on role-players and skimp on stars.

This ``elites first" strategy was used by the Boston Celtics in the summer of 2007 when they traded their role-players and other assets to build a team around Kevin Garnett, Paul Pierce and Ray Allen\footnote{\url{http://www.nba.com/celtics/news/press073107-garnett.html}}, and more recently
by the Miami Heat in the summer of 2010 who built a team around LeBron James, Dywane Wade and Chris Bosh\footnote{\url{http://sports.espn.go.com/nba/news/story?id=5365165}}. We shall incorporate this prior information through $\ell_1$ regularization. This penalizes non-sparse models, and should cause only the very best players to stand out in the regression.
This suggests a penalty term of the form
$\lambda_1 ||\bbeta ||_1$.

\item Box score information: Another valuable piece of information useful in inferring player worth is the box score statistics matrix $\Rmat$.
One expects good players to not only have high APM ratings, but to also produce rebounds, assists, blocks, steals, etc. Thus, we prefer
ratings $\hat{\bbeta}$ which are consistent with box score statistics.
In other words, we expect a ratings vector to be ``close" to the column
space of $\Rmat$. We therefore should penalize ratings for which
the distance from $\hat{\bbeta}$ to  $\Rmat \bz$ is large. Although there are many different possible penalties one can choose, in this work we choose
a quadratic penalty term of the form $\lambda_2 ||\bbeta - z_0\ones_\pdim - \Rmat \bz||_2^2$.
\end{enumerate}

We can encode the above prior information through the function
$g(\bhca, \bbeta, z_0, \bz; \lambdavec)$ defined as
\begin{align}
\label{lambdapm_objective_function}
g(\bhca, \bbeta, z_0, \bz; \lambdavec)
 \defn \underbrace{\weightedloss}_{\text{Weighted least squares}} + 
{{{\underbrace{\lambda_1 ||\bbeta ||_1}_{\text{Sparse player ratings}}}} } + 
{{{\underbrace{\lambda_2 ||\bbeta - z_0\ones_\pdim - \Rmat \bz||_2^2}_{\text{Box score prior}}}}},
\end{align}

\begin{enumerate}
\item $\Rmat$ is a $\pdim \times \ddim$ matrix containing the box-score statistics of the $\pdim$ different players,
\item The variable $\bz$ gives us weights for each of the box score statistics,
\item and the vector $(\lambda_1, \lambda_2) \in \realpos^2$ are the regularization parameters.
\end{enumerate}

We shall use the shorthand $\lambdavec$ to denote the pair $(\lambda_1, \lambda_2)$.
We can find a model consistent with both the data and the prior information by solving the convex optimization problem
\begin{align}
\label{my_cvx}
\betalamhca, \betalam, \zzerolam, \zlam = \arg\min g(\bhca, \bbeta, z_0, \bz; \lambdavec).
\end{align}

We call the procedure described by Equation \eqref{my_cvx} the $\lambdapm$ algorithm, and the vector $\betalam$ are the player ratings produced by it.
One very important difference between $\lambdapm$ and the least squares approach is that it yields both a player
rating vector $\betalam$ and a box score weights vector $\zlam$.
The weights vector $\zlam$ is a valuable tool in its own right. It provides numerical values for different basic box score statistics like scoring, rebounding, and steals.
We further interpret
$\zlam$ in Section \ref{bsweights_interpretation}.

Furthermore, it yields a linear formula for transforming
player box score effectiveness into a player productivity rating through the equation

\begin{align}
\label{bsrating}
\balpha \defn \Rmat \zlam + \zzerolam\ones_\pdim.
\end{align}

$\balpha$ can be viewed as an additional player rating vector produced by $\lambdapm$, one that linearly
transforms each player's box score production into a ``points per 100 posession" rating similar to least squares or $\betalam$. Thus,
$\balpha$ succintly converts the box score production of each player into a single number.
 
 Thus for player $i$, we can compare the variable $\betalam_i$ to the variable $\balpha_i$ to understand how ``overrated" or underrated he is relative to his box score production. This is useful, since many players produce great box score statistics but don't necessarily impact team competitiveness to the level the box score might suggest. We explore this aspect of $\lambdapm$ in further detail in Section \ref{underrated_analysis}.

\subsection{Bayesian Interpretion of $\lambdapm$}
$\lambdapm$ can be interpreted as the posterior mode for a Bayesian statistics model.
Suppose that $\Yvec, \bhca, \bbeta, z_0, \bz$ are all random variables.

Let
\bit
\item $\Yvec_i | \bhca, \bbeta \sim \mathcal{N}(\bhca + \Xmat_i^T \bbeta, \frac{1}{2})$,
\item $\bhca$ have the improper prior $\prob(\bhca=\alpha) \propto 1$
\item $\prob(\bbeta | z_0, \bz) \propto e^{-\lambda_1 ||\bbeta ||_1 - \lambda_2 ||\bbeta - z_0 \ones_\pdim - \Rmat \bz ||_2^2}$,
\item $z_0$ have the improper prior $\prob(z_0=\gamma) \propto 1$,
\item and $\bz$ has the improper prior $\prob(\bz=\kappa) \propto 1$.
\eit

Then the solution to $\lambdapm$ with $\wvec = \ones_\numobs$ is exactly the mode
of the posterior distribution $\prob(\bhca, \bbeta, z_0, \bz| \Yvec)$.

\subsection{Selecting the regularization parameter $\lambdavec$}
\label{appendix_cv}
For $\lambdapm$ to be useful, we need to be able to select a good choice of $\lambdavec$ quickly.
Cross-validation \citep{Stone1974} is one standard technique in statistics for doing this. To select regularization parameters, we use $10$-fold cross-validation. We cross-validate
over regularization parameters from the set
\begin{align*}
\Lambda := \{ (2^{a}, 2^{b}) | \text{ } a,b \in F \}
\end{align*}

where
\begin{align*}
F \defn \{-10, -9, \ldots, 9\}.
\end{align*}

\begin{table}
  \caption{Regularization parameters obtained from 10-fold cross-validation}
\centering
\begin{tabular}{l*{3}{c}r}
              Setting & $\lambda_1$ & $\lambda_2$ \\
\hline
\DTLforeach{lamvals}{
\firstname=setting,\foo=lam1,\surname=lam2}{%
\firstname & \foo & \surname \\}
\end{tabular}
\label{tab:lamvals}
\end{table}

$K$-fold cross-validation on $T$ different values of $\lambdavec$ means solving $T K$ different $\lambdapm$ problems, each
of which are convex programs of moderate size ($\numobs \approx 20000$, $\pdim \approx 450$, $\ddim \approx  20$).

Thus, it is necessary that
\begin{enumerate}
\item for each fixed value of $\lambdavec$, $\lambdapm$ can be solved quickly
\item  and that many values of $\lambdavec$ can be evaluated at once.
\end{enumerate}

To address the first issue, we implemented a fast numerical algorithm for solving
$\lambdapm$ for a fixed valued of $\lambdavec$. See Appendix \ref{appendix_cyclical}
for a derivation.

To address the second issue, our cross-validation code takes advantage of the cloud computing service PiCloud\footnote{\url{http://www.picloud.com}} to perform the computations in parallel.

The resulting regularization parameters learned by cross-validation are summarized in Table \ref{tab:lamvals}.

\section{The Performance of $\lambdapm$}
Our ultimate goal is to produce substantially better estimates of player value than least squares. If it turns out that despite all the additional computational work that $\lambdapm$ requires that there is little or no statistical improvement, then $\lambdapm$ is not of much practical value.
In this section, we discuss the performance of $\lambdapm$ on the 2010-2011 NBA dataset.  
We demonstrate that $\lambdapm$ substantially outperforms both the dummy estimate and least squares estmate, and outperforms even
Las Vegas given a sufficient amount of training data.

\subsection{$\lambdapm$ outperforms least squares}
\label{lamp_rocks_apm}
From Table \ref{tab:lamvals}, we see that the cross-validation methodology described in Section \ref{appendix_cv} on the first 820 games of the 2010-2011 season yields the regularization parameter
\begin{align*}
\lambdavec^{820}_{CV} = (2^{-10}, 2^{-3}).
\end{align*}

Armed with this choice, we can now compare least squares to $\lambdapm$ on the final 410 games of the 2010-2011 NBA regular season. Each procedure produces a player rating vector $\betahat$, and we can use these ratings to predict the final margin of victory over this collection of games. 

Recall that we use the variable $\hat{A}_i$ to denote the number of points that a statistical estimator predicts that the home team will win game $i$, $A_i$ to denote the \emph{actual} number of points by which the home team wins game $i$, and $\hat{E_k} = \hat{A}_i - A_i$ to denote the difference between these quantities.

Figure \ref{figdump} is a histogram of the variable $\hat{E_k}$ for each technique.
It is clear from Figure \ref{figdump} that $\lambdapm$ produces better estimates than APM. The histogram of the $\lambdapm$ errors are ``spikier" around the origin than the APM errors. We can also study some of the empirical properties of the variable $\hat{E_k}$ for each approach. Table \ref{tab:hca_vs_apm} summarizes the results. 
As Table \ref{tab:hca_vs_apm} indicates, $\lambdapm$ represents a substantial improvement on APM in nearly all of these statistical measures.
In particular,
\bi
\item the fraction of games in which the wrong winner is guessed decreases from 33.66\% with LS to 28.54\% with \lambdapm; and
\item the average absolute error in predicting the margin of victory decreases from 18.05 to 10.554.
\ei

Comparing $\lambdapm$ to the dummy estimator, we
\bi
\item see an enormous improvement in ability to predict the winning team. The percentage of games in which the wrong winner is predicted falls from $39.27\%$ to $28.54\%$.
\item Both techniques obtain a similar average absolute error in predicting the margin of victory, with $10.54$ for the dummy estimator and $10.55$ with \lambdapm.
\ei

Overall, this suggests that $\lambdapm$ more accurately models the NBA than the least squares estimator.

\subsection{$\lambdapm$ outperforms Las Vegas}
\label{lamp_rocks_vegas}

To convincingly evaluate the performance of $\lambdapm$, we examine whether it actually results in a profitable gambling strategy against the Vegas lines.
In fact, we will compare the dummy, least squares and $\lambdapm$ estimators.
Given predictions by each of the above estimates, we have the following natural gambling strategy:
\begin{enumerate}
\item If the deviation $\Delta$ between the estimate's prediction of the outcome of a game and the Vegas lines is greater than 3, place
a bet on the team the estimator favors.
\item Otherwise, don't bet.
\end{enumerate}

Due to transaction costs that the sportbooking companies charge\footnote{The fee is called the ``vigorish" in the gambling community.} a gambling strategy must win more than roughly $52.5\%$ of the time to at least break even.
Table \ref{tab:gamble820} summarizes the result of this gambling rule for each of the three techniques of interest over the last $410$ games of the 2010-2011 NBA season.
The dummy-based gambling strategy places 263 bets on the 410 games and loses 3 more bets than it wins, for a winning percentage below 50\%, which is performance comparable to random guessing, and not enough to break even.
The least squares-based strategy has a winning percentage of 51.97\% on 356 bets made.
In comparison, $\lambdapm$ places wagers on 290 games and wins $57.24\%$ of these bets. This represents a very profitable betting strategy, and thus suggests that $\lambdapm$ more accurately models the NBA than major alternatives, including the estimators used by Las Vegas.
Finally, $\lambdapm$ obtains this improved performance while only having access to the first $820$ games of the regular season.

\begin{table}[!ht]
  \centering
  \begin{threeparttable}
    \caption{Betting Strategy over the last 410 games, $\Delta=3$}
      \begin{tabular}{l*{5}{c}r}
                    Statistic & Dummy & LS & RR & \lambdapm & \lambdapm 2 \\
      \hline
      \DTLforeach{gamble820}{%
      \firstname=metric,\foo=Dummy,\surname=WLS,\pore=Bamp-CV,\core=Bamp-CV2,\krr=Ridge}{%
      \firstname & \foo & \surname & \krr &  \pore & \core \\}
      \end{tabular}
  \end{threeparttable}
\label{tab:gamble820}
\end{table}

\subsection{Robustness of Results}
How sensitive is the $\lambdapm$ algorithm to our choice of training on the first 410 games? Does
the performance relative to the least squares estimate degrade if the estimators are trained on much fewer games?
To evaluate this, we train estimators on the first 410 games and then evaluate predictive power on the remaining 820 games.
From Table \ref{tab:lamvals}, we obtain the cross-validation selected regularization parameter 

\begin{align*}
\lambdavec^{410}_{CV} = (2^{-7}, 2^{-2}).
\end{align*}

We also compare against the Las Vegas predictions for that block of 820 games. Table \ref{tab:robust1} summarizes the results of this experiment. As before, $\lambdapm$ outperforms
both the Dummy estimator and LS.
Furthermore, by increasing $\Delta$ to $5$ (from the value $3$ used when training
on 820 games), $\lambdapm$ still leads to a successful betting strategy,
as Table \ref{tab:gamble410} shows.

\begin{table}[!ht]
\centering
  \begin{threeparttable}
  \caption{Robustness Experiment, First 410 Games}
    \begin{tabular}{l*{5}{c}r}
                  Metric & Dummy & LS & RR & \lambdapm \\
    \hline
    \DTLforeach{robust205}{%
    \firstname=metric,\surname=Dummy,\score=WLS,\pore=Bamp-CV,\bore=Ridge}{%
    \firstname& \surname & \score & \bore & \pore   \\}
    \end{tabular}
  \label{tab:robust1}
  \end{threeparttable}
\end{table}

\begin{table}[!ht]
  \centering
  \begin{threeparttable}
    \caption{Betting Strategy over the last 820 games, $\Delta=5$}
      \begin{tabular}{l*{5}{c}r}
                    Statistic & Dummy & LS & RR & \lambdapm \\
      \hline
      \DTLforeach{gamble410}{%
      \firstname=metric,\foo=Dummy,\surname=WLS,\pore=Bamp-CV,\krr=Ridge}{%
      \firstname & \foo & \surname & \krr &  \pore \\}
      \end{tabular}
  \label{tab:gamble410}
  \end{threeparttable}
\end{table}

\section{What does $\lambdapm$ say about the NBA?}
In the previous section, we evaluated the performance of $\lambdapm$ by testing its ability to predict the outcome of unseen games. In this section, we interpret the box score weights vector $\zlam$ and player rating vector $\betalam$ returned by $\lambdapm$, and discuss what they say about the NBA.

\subsection{Top 10 players in the league}
\label{best_players}

\begin{table}[!ht]
  \centering
  \begin{threeparttable}
    \caption{$\lambdapm$ Player Ratings}

      \begin{tabular}{l*{5}{c}r}
                    Player & $\betalam_{i}$\\
        \hline
        \DTLforeach{bamptop10}{%
        \surname=Name,\score=BAMPCV}{%
        \surname & \score\\}
      \end{tabular}
    \label{tab:my_table_bamp}
  \end{threeparttable}
\end{table}

From $\betalam$ we can extract a list of the top 10 players in the league who have played at least 10 possessions. Table \ref{tab:my_table_bamp} summarizes these results.
This list contains some of the most prominent star players in the league (LeBron James, Chris Paul, Dirk Nowitzki, Dwight Howard), thus agreeing with common basketball wisdom.
However, this ranking contradicts common basketball wisdom in the following ways:
\begin{enumerate}
\item The list noticeably omits Kobe Bryant, a player pop culture and common basketball wisdom considers one of the league's superstars. Yet $\lambdapm$ thinks very highly of Pau Gasol and Lamar Odom, two of Kobe Bryant's teammates who are individually credited far less for the success of the Lakers than Kobe is. 
\item The list includes Nene Hilario and Jeremy Evans, players who are not considered by most to be amongst the top 10 players in the league.
\end{enumerate}

\subsection{Top 10 most underrated and overrated players}
\label{underrated_analysis}
There are certain players in the NBA for whom their impact on the game seems to be far more (or less) than their raw box score production suggests.
$\lambdapm$ allows us to identify these players and quantify their impact by measuring the discrepancy between their $\lambdapm$ rating and their weighted box score ratings $\balpha_i$.

We define the underrated vector $\Uvec$ as
\begin{align*}
\Uvec \defn \betalam - \balpha.
\end{align*}.

Similarly, we can examine which players impact the game much less than their box score production suggests with the vector $\Ovec \defn -\Uvec$.

Table \ref{tab:underratedplayers} lists the top 10 most underrated/overrated players in the league relative to their box score production.
For at least a few of these players, it is easy to understand why box scores alone do a poor job of capturing their impact:
\begin{itemize}
\item Andris Biedrens is a severe liability offensively, due to both his inability to score outside of 5 feet of the basket and poor free throw shooting. This makes it much more difficult for his teammates to score, since his defender
can shift attention away from him and instead provide help elsewhere. Biedrens is also a liability defensively.
\item Goran Dragic is a point guard with a scoring mentality. While a ``shoot-first" point guard is not necessarily
harmful to a team, if he doesn't do a good enough job in setting up his teammates and creating easy scoring opportunities for them, it hurts his team's ability to score.
\end{itemize}

\subsection{Box score weights produced by $\lambdapm$}
\label{bsweights_interpretation}
The $\lambdapm$ regression also produces box score weights $\zlam$ that tell us the relative importance of the different box score statistics.
$\zlam$ gives us a method to linearly transform box score data into the player effectiveness rating $\balpha$ defined in Equation \ref{bsrating}.
For player $j$, the variable $\balpha_j$ is a weighted linear combination of his box score statistics.

We can examine each entry of the vector $\zlam$ to compare the relative importance of different box score
variables like rebounds, assists and steals.
Table \ref{tab:bsweights} summarizes the results. We also display the relevant row in the box score matrix $\Rmat$ for LeBron James, which we call $\Rmat_{\text{Lebron James}}^T$.

Examining this table, we see that LeBron James made two point shots at a rate of $7.83$ per 36 minutes, and attempted two point shots at a rate of $14.16$ per 36 minutes. The corresponding weightings from $\zlamrescaled$ are $3.38$ and $-1.54$ respectively, suggesting that overall LeBron's rating from his two point shooting is $7.83 \times 3.38+14.16\times-1.54 \approx 4.66$ points. In fact, from these weightings we can calculate that according to the $\lambdapm$ model all players in the league must hit their two point shots roughly $45\%$ of the time for their rating from two point shooting to be non-negative.

Interestingly enough, a similar calculation reveals that three point shots must only be hit at a roughly $14\%$ rate to break even. This is counterintuitive: naively one would believe that hitting two point shots $q$ percent of the time should be equivalent to hitting three point shots $\frac{2}{3} q$ of the time.
However, three point shooting increases the amount of spacing on the floor and perhaps missed
three point shots are easier to rebound for the offensive team.

According to this interpretation of the $\zlam$ variable turnovers are extremely costly, with the corresponding entry of $\zlamrescaled$ equal to $-.76$. Thus, LeBron's turnover rate of $3.34$ turnovers per 36 minutes hurts his rating box score rating by roughly $6.28$ points.

\begin{table}[!ht]
\centering
  \begin{threeparttable}
  \caption{Box Score Weights}
\begin{tabular*}{0.75\textwidth}{ l*{4}{c}r }
Statistic & Description	& $\zlam$  & $\Rmat_{\text{LeBron James}}^T$\\
\hline
\DTLforeach{weights1}{%
\s=Statistic,\d=Descriptions,\w=Weight,\l=LeBron}{%
\s & \d & \w & \l \\}
\end{tabular*}
  \label{tab:bsweights}
  \end{threeparttable}
\end{table}

\begin{table}[!ht]
  \centering
  \begin{threeparttable}
  \caption{Underrated/Overrated Players}
\small
{
	\begin{tabular}{l*{3}{c}r}
		      Player & $\betalam$ & $\Rmat \zlam + \zzerolam \ones_\pdim$ & Underrated \\ \hline
	\DTLforeach{underrated}{%
	\firstname=Name,\surname=BAMPCV,\score=BAMPCVBoxScore,\scorea=Discrepancy}{%
	\firstname & \surname & \score & \scorea\\}
	\end{tabular}
}
{
	\begin{tabular}{l*{3}{c}r}
		      Player & $\betalam$ & $\Rmat \zlam + \zzerolam \ones_\pdim$ & Overrated\\ \hline
	\DTLforeach{overrated}{%
	\firstname=Name,\surname=BAMPCV,\score=BAMPCVBoxScore,\scorea=Discrepancy}{%
	\firstname & \surname & \score & \scorea\\}
	\end{tabular}
}
  \label{tab:underratedplayers}
  \end{threeparttable}
\end{table}

\section{Extending $\lambdapm$ by augmenting the box score}
\label{future}
In this section, we discuss a possible extension to the $\lambdapm$ model.

The box score matrix $\Rmat$ keeps track of statistics like rebounds, assists, and steals.
However, one might imagine augmenting this basic box score matrix with products of raw statistics such as rebounds $\times$ assists, blocks $\times$ steals, turnovers $\times$ free throws made, etc.
By capturing some of these product statistics and incorporating them into $\lambdapm$, one might more accurately model the value of multifaceted players.

We expand the matrix $\Rmat$ to include all pairwise product of the basic variables.
If $\Rmat$ is a $\pdim$ by $\ddim$ matrix, this leads to a $\pdim$ by $\ddim + \binom{\ddim}{2}$ matrix called

\begin{align*}
\Rmatexp.
\end{align*}

Let us use the notations $\lambdapm(\Rmat)$ and $\lambdapm(\Rmat, 2)$ to denote $\lambdapm$ with the box score matrices
$\Rmat$ and $\Rmatexp$, respectively.
Applying the cross-validation procedure described in Section \ref{appendix_cv} on the first 820 games of the 2010-2011
produces the regularization parameter

\begin{align*}
\lambdavec^{820}_{CV, 2} = (2^{-10}, 2^{-1}).
\end{align*}

With this choice of parameter for the expanded box score matrix $\Rmatexp$, we can then empirically
compare its performance to that of the ordinary $\lambdapm$ algorithm using the basic box score matrix $\Rmat$.
Figure \ref{figdump} demonstrates the result of this experiment.
From this figure we see that the additional box score statistics don't seem to substantially improve performance.
The histogram of the $\lambdapm(\Rmat, 2)$ errors are fairly similar to the $\lambdapm(\Rmat)$ errors.
We can also study some of the empirical properties of the variable $\hat{E_k}$ for each approach. Table \ref{tab:hca_vs_apm}
summarizes the results. 

As Table \ref{tab:hca_vs_apm} indicates, $\lambdapm(\Rmat, 2)$ doesn't improve upon the predictive power of $\lambdapm(\Rmat)$.
The fraction of games in which the wrong
winner is guessed actually increases from $28.54\%$ to $29.51\%$, the average absolute error in predicting games
increases from 10.55 to 11.81.

A possible explanation for this poor statistical performance is that the pairwise
interaction terms that $\lambdapm(\Rmat, 2)$ models are too many, and thus the model
is overfitting.

\section{Conclusion}
We have introduced $\lambdapm$, a powerful new statistical inference procedure for the NBA. We compared the statistical performance of our approach to an existing popular technique based on least squares and demonstrate empirically
that $\lambdapm$ gives more predictive power.
We also compare $\lambdapm$ to the Las Vegas lines and show that with sufficient training data, $\lambdapm$ seems to better predict the NBA than Vegas.
We interpret the estimates produced by $\lambdapm$ and discuss what they suggest about who the best players in the NBA are,
and which players are overrated or underrated.
Finally, we discuss a possible extension to the $\lambdapm$ model.

\newpage

\appendix
\section{The Cyclical Coordinate Descent Algorithm for $\lambdapm$}
\label{appendix_cyclical}
There are a variety of techniques for solving the convex program \eqref{my_cvx}, including
interior-point methods \citep{Boyd02}, LARs \citep{efron04}, iteratively re-weighted least squares \citep{huber:1974},
approximating the $\ell_1$ term with a smooth function \citep{LeeLeeAbbNg06}
the sub-gradient method \citep{shor1985}, and Nesterov's proximal gradient method \citep{nesterov2004}.

Ultimately, we found experimentally that cyclical coordinate descent (CCD) \citep{friedman_cyclical, wu_cyclical} was the fastest for our problem.

The CCD method works by repeatedly optimizing the objective function viewed as a function of each variable with the others fixed. This idea gives a CCD algorithm for $\lambdapm$, Algorithm \ref{alg1}.

\begin{algorithm}
\caption{CCD$\lambdapm$($\Xmat$, $\Yvec$, $\Rmat$, \lambdavec, T)}
\label{alg1}
\begin{algorithmic}[1]
\STATE $\bhca(0) \gets 0$, $z_0(0) \gets 0$, $\bbeta(0) \gets \zeroes_{\pdim}$, $\bz(0) \gets \zeroes_{\ddim}$
\FOR{$i \in \{1, 2, \ldots, T\}$}
  \STATE \COMMENT{Optimize $\bhca$ with all other variables fixed}
  \STATE \COMMENT{Optimize $z_0$ with all other variables fixed}
  \FOR{$k \in \{1, 2, \ldots, \pdim\}$}
  \STATE \COMMENT{Optimize $\bbeta_k$ with all other variables fixed}
  \ENDFOR
  \FOR{$\ell \in \{1, 2, \ldots, \ddim\}$}
  \STATE \COMMENT{Optimize $\bz_\ell$ with all other variables fixed}
  \ENDFOR
\ENDFOR
\RETURN $\bhca(T), z_0(T), \bbeta(T), \bz(T)$
\end{algorithmic}
\end{algorithm}

\subsection{Convergence of Algorithm \ref{alg1}}
The correctness of this algorithm for minimizing the objective function \eqref{lambdapm_objective_function} follows from Lemma \ref{alg1conv}.

\begin{lem}
\label{alg1conv}
Let $\bhca^*, z_0^*, \bbeta^*, \bz^* \in \argmin g(\bhca, \bbeta, z_0, \bz; \lambdavec)$.

Then
\begin{align*}
\lim_{T \to \infty} g(\bhca(T), z_0(T), \bbeta(T), \bz(T)) = g(\bhca^*, z_0^*, \bbeta^*, \bz^*).
\end{align*}

Furthermore, when $\bhca^*, z_0^*, \bbeta^*, \bz^*$ is the unique global
minimum of $g$,
\begin{align*}
\lim_{T \to \infty} (\bhca(T), z_0(T), \bbeta(T), \bz(T)) = (\bhca^*, z_0^*, \bbeta^*, \bz^*).
\end{align*}
\end{lem}

\begin{proof}
This is a direct consequence of Proposition 5.1 of \cite{Tsen01a}.
In particular, identify $f_0$ and $f_i, i=1, \ldots,\pdim$ of Proposition 5.1 with $\weightedloss + 
\lambda_2 ||\bbeta - z_0\ones_\pdim - \Rmat \bz||_2^2$ and $\lambda_1 |\bbeta_i|, i=1, \ldots, \pdim$, respectively. We observe that
\bit
\item
Assumption B1 of \cite{Tsen01a} is satisfied, since $f_0$ is continuous.
\item Assumption B2 of \cite{Tsen01a} is satisfied, since $f$ is convex
and non-constant on line segments.
\item
Assumption B3 is satisfied, since $f_i, i=1,\ldots, \pdim$ are continuous.
\item Assumption C2 is trivially satisfied.
\eit
Therefore, the conditions of Proposition 5.1 of \cite{Tsen01a}
are satisfied for Algorithm \ref{alg1} on the objective
function \eqref{apm_algorithm}.

Since \eqref{lambdapm_objective_function} has at least one global minimum and is convex, then we further conclude that the limit points of Algorithm \ref{alg1}
are global minima.

\end{proof}

\subsection{Computing the updates for Algorithm \ref{alg1}}
The updates for $\bhca(i), z_0(i), \bbeta(i), \bz(i)$ can be computed in closed form.

To compute $\bhca(i)$, we can optimize the objective function $g$ viewed as a function only
of the decision variable $\bhca$ by taking the derivative and setting it to zero.

This yields the update

\begin{align*}
\bhca \gets  \frac{1_{\numobs}^T \text{Diag}(w) \left[\Yvec - \Xmat \beta \right]}{1_{\numobs}^T w}.
\end{align*}

Similarly, for $z_0(i)$ we get the update
\begin{align*}
z_0 \gets  \frac{1}{\pdim} \ones_\pdim^T \left[\beta - \Rmat z \right].
\end{align*}

For $\bz$, we simply get the least squares updates:
\begin{align*}
\bz \gets (\Rmat^T\Rmat)^{-1} \Rmat^T [\bbeta - z_0 \ones_p].
\end{align*}

\subsubsection{Updates for $\bbeta$}
We next derive a closed-form expression for the updates for $\bbeta_i$. To do so, we need Lemma \ref{soft_thresh_lem}.

\begin{lem}[One-variable lasso is soft-thresholding]
\label{soft_thresh_lem}
Let $h(x) := \frac{1}{2}A x^2 - Bx + C + \tau |x|,  x \in \real$.
Suppose that $A>0$.
The solution of
\begin{align}
\label{softthresh}
\min_{x \in \real} h(x)
\end{align}

is

\begin{align*}
x^* = \frac{\shrink{\tau}(B)}{A}
\end{align*}

where
\begin{align*}
\shrink{\tau}(x) &\defn
\begin{cases}
0 &\mbox{if } |x| \leq \tau \\ 
x-\tau &\mbox{if } x > 0 \text{ and } |x| > \tau \\
x+ \tau &\mbox{if } x < 0 \text{ and } |x| > \tau,
\end{cases} 
\end{align*}

is the soft-thresholding function with threshold $\tau$.
\end{lem}

See Section \ref{proof_soft_thresh_lem} for a proof of this.

This lemma is useful, as it allows us to immediately write the updates for $\bbeta_k(i)$.

First, let us identify $A$ and $B$ for $\bbeta_k(i)$. Differentiating $g_s$, we get
\begin{align*}
\partial_{\bbeta_t} g_s &= 
\partial_{\bbeta_t} ( \weightedloss + \lambda_2 ||\bbeta - z_0 \ones_\pdim- \Rmat \bz||_2^2) \\
&= \partial_{\bbeta_t} \left[\frac{1}{\ones_n^T \vec{w}} \sum_{i} w_i (Y_i - \bhca - \Xmat_i^T \bbeta)^2\right]
+ \lambda_2 \partial_{\bbeta_t}  \sum_{j} (\bbeta_j - z_0 - \Rmat_j^T z)^2 \\
&= ( \frac{1}{\ones_n^T \vec{w}}  \sum_{i} w_i \partial_{\bbeta_t} (Y_i - \bhca - \Xmat_i^T \bbeta)^2
+ \lambda_2 \sum_{j} \partial_{\bbeta_t} (\bbeta_j - z_0 - \Rmat_j^T z)^2 \\
&= \frac{1}{\ones_n^T \vec{w}}  \sum_{i} 2 w_i \Xmat_{it} (-Y_i + \bhca + \Xmat_i^T \bbeta)
+ \lambda_2 2(\bbeta_t - z_0 - \Rmat_t^T z) \\
&= C  \sum_{i} w_i \Xmat_{it} (-Y_i + \bhca + \Xmat_i^T \bbeta)
+ \lambda_2 2(\bbeta_t - z_0 - \Rmat_t^T z) \\
&= C  (\Xmat e_t)^T W (-Y + \bhca \ones_n + \Xmat \bbeta)
+ \lambda_2 2(\bbeta_t - z_0 - \Rmat_t^T z) \\
&= C  (\Xmat e_t)^T W (-Y + \bhca \ones_n + \Xmat [\bbeta-e_t \bbeta_t + e_t \bbeta_t]
+ \lambda_2 2(\bbeta_t - \theta_t) \\
&= C  (\Xmat e_t)^T W (\kappa + \Xmat [e_t \bbeta_t])
+ \lambda_2 2(\bbeta_t - \theta_t)
\end{align*}

where
\begin{align*}
C &:= \frac{2}{\ones_n^T \vec{w}},\\
\theta_t &:= (z_0 \ones_p + \Rmat \bz)^T e_t,\\
\kappa &:=-Y + \bhca \ones_n + \Xmat [\bbeta-e_t \bbeta_t].
\end{align*}

The constant term (with respect to $\bbeta_t$) of the above expression is
\begin{align*}
D := C  (\Xmat e_t)^T W \kappa
- 2\lambda_2 \theta_t.
\end{align*}

The linear term is
\begin{align*}
C  (\Xmat e_t)^T W \Xmat e_t \bbeta_t
+ 2 \lambda_2 \bbeta_t &= \left[C  e_t^T \Xmat^T W \Xmat e_t
+ 2 \lambda_2\right] \bbeta_t \\
&= E \bbeta_t,
\end{align*}

where
\begin{align*}
E := C  e_t^T \Xmat^T W \Xmat e_t + 2 \lambda_2.
\end{align*}

From this we conclude that for $\bbeta_k(t)$
\begin{align*}
A = E \\
B := -D.
\end{align*}

So, we have the update equation
\begin{align*}
\bbeta_k(i) &\gets \frac{\shrink{\lambda_1}(B)}{A}.
\end{align*}

\subsubsection{Proof of Lemma \ref{soft_thresh_lem}}
\label{proof_soft_thresh_lem}
\begin{proof}
The subdifferential of $h(x)$ \citep{Rockafellar} is the set
\begin{align*}
\partial{h(x)} &:= \sum_{k=1}^K a_k (a_k x - b_k) + \tau \partial{|x|} \\
&= x A - B  + \tau \partial{|x|},
\end{align*}

where
\begin{align*}
A &\defn \sum_{k=1}^K a_k^2 \\
B &\defn \sum_{k=1}^K b_k \\
\partial{|x|} &\defn \begin{cases}
                        \{\sign(x)\} &\mbox{if } x \neq 0 \\ 
                       [-1, 1] &\text{otherwise.}
                    \end{cases}
\end{align*}

From the theory of convex analysis \citep{Rockafellar} $x^*$ is the solution of \eqref{softthresh} if and only if
\begin{align}
\label{subg_opt}
0 \in \partial{h(x^*)}.
\end{align}

The set $\partial{h(x^*)}$ behaves differently depending on the value of $x^*$. When $x^* \neq 0$, then
\begin{align*}
0 \in \partial{h(x^*)} = \{x^* A - B  + \tau \sign(x^*)\},
\end{align*}

which is equivalent to $x^* =\frac{B - \tau \sign(x^*)}{A}$. However, when $x^*=0$, then
\begin{align*}
0 \in \partial{h(x^*)} = \{x^* A - B  + \tau [-1,1]\}.
\end{align*}

We use this observation to deal with the following two cases:
\begin{enumerate}
\item Suppose that $\tau \geq |B|$. Then
\begin{enumerate}
  \item
      If $x^* \neq 0$, then
      \begin{align*}
      \tau &\geq |B| = |x^* A + \tau \sign(x^*)| \\
                       &= x^* A + \tau,
      \end{align*}
      since at least one $a_k \neq 0$, then $A >0$. This is a contradiction. Therefore $x^* \neq 0$ cannot be a solution
      when $\tau \geq |B|$.
  \item
      If $x^* = 0$, then
      \begin{align*}
        0 \in -B + \tau [-1, 1] = [-\tau-B, \tau-B],
      \end{align*}
      which is true.

\end{enumerate}
\item Suppose that $\tau < |B|$. Then
\begin{enumerate}
  \item
        If $x^* \neq 0$, then $B-\tau\sign(x)$ has the same sign as $B$. Since $A$ is positive,
        then $x^*$ has the same sign as $B$. So the choice of $x^* = \frac{B - \tau\sign(B)}{A}$ satisfies the required
sub-gradient optimality condition \eqref{subg_opt} without contradiction.
  \item
        If $x^* = 0$, then $0 \in -B + \tau [-1,1] = [-\tau-B, \tau-B]$, which is a contradiction. 
\end{enumerate}
\end{enumerate}

Thus,
\begin{enumerate}
\item $\tau \geq |B| \implies x^* =0$,
\item $\tau < |B| \implies x^* =\frac{\shrink{\tau}(B)}{A}$.
\end{enumerate}

These two cases can be summarized by $x^* = \frac{\shrink{\tau}(B)}{A}$, as desired.
\end{proof}

\bibliographystyle{plainnat}
\bibliography{new_dapo}
\end{document}